\documentclass[11pt]{article}

\usepackage{booktabs}
\usepackage{wrapfig}
\usepackage[numbers]{natbib}

\usepackage{amsfonts, amsthm, amsmath, comment}
\usepackage[boxed]{algorithm2e}
\usepackage{verbatim}
\usepackage{amsfonts}

\usepackage{url}
\usepackage{float}
\usepackage[margin=1in]{geometry}
\usepackage[usenames]{color}
\usepackage[dvipsnames,usenames,table]{xcolor}
\usepackage[colorlinks=true,urlcolor=blue,linkcolor=RoyalBlue,citecolor=OliveGreen]{hyperref}
\usepackage{tikz}
\usetikzlibrary{matrix}
\usepackage{pgfplots}
\pgfplotsset{compat=1.13}
\usepackage{graphicx}
\usepackage{caption}
\usepackage{subcaption}
\usepackage{multirow}
\usepackage{thmtools}
\usepackage{thm-restate}
\usepackage{cleveref}

\newtheorem{lemma}{Lemma}
\newtheorem{corollary}{Corollary}
\newtheorem{claim}{Claim}
\newtheorem{example}{Example}

\newcommand{\ut}[2]{u\left( #1 \rightarrow #2 \right)}
\newcommand{\w}[2][]{w_{#1} \left( #2 \right)}
\newcommand{\p}[1]{p\left( #1 \right)}
\newcommand{\pb}{\mathbf{p}} 
\newcommand{\vstar}[1]{v^*_#1}
\newcommand{\defeq}{\stackrel{\triangle}{=} }
\renewcommand{\path}{\mathcal{P}}

\newcommand{\icrit}{i^*}
\newcommand{\iend}{\hat{i}}
\newcommand{\total}{k}

\newcommand{\integers}{\mathbb{Z}}
\newcommand{\reals}{\mathbb{R}}

\DeclareMathOperator*{\E}{\mathbb{E}}
\DeclareMathOperator*{\Prob}{\mathbb{P}}
\DeclareMathOperator{\Rev}{Rev}

\newcommand{\anote}[1]{{\color{red} {(\sf Alex's Note:} {\sl{#1}}     {\sf )}}}

\newenvironment{talign}
 {\let\displaystyle\textstyle\align}
 {\endalign}
\newenvironment{talign*}
 {\let\displaystyle\textstyle\csname align*\endcsname}
 {\endalign}

\title{Optimal Multi-Unit Mechanisms with Private Demands}

\author{Nikhil R. Devanur 
\\ Microsoft Research 
\and
Nima Haghpanah
\\ Penn State University
\and
Christos-Alexandros Psomas
\\UC Berkeley
}

\date{}
\begin{document}

\maketitle

\begin{abstract}
		In the multi-unit pricing problem, multiple units of a single item are for sale. A buyer's valuation for $n$ units of the item is $v \min \{ n, d\} $, where the per unit valuation $v$ and the capacity $d$ are private information of the buyer. We consider this problem in the Bayesian setting, where the pair $(v,d)$ is drawn jointly from a given probability distribution.  In the \emph{unlimited supply } setting, the optimal (revenue maximizing) mechanism is a pricing problem, i.e., it is a menu of lotteries. In this paper we show that under a natural regularity condition on the probability distributions, which we call {\emph{decreasing marginal revenue}}, the optimal pricing is in fact \emph{deterministic}. It is a price curve,
		offering $i$ units of the item for a price of $p_i$, for every integer $i$. Further, we show that the revenue as a function of the prices $p_i$ is a \emph{concave} function, which implies that the optimum price curve can be found in polynomial time. This gives a rare example of a natural multi-parameter setting where we can show such a clean characterization of the optimal mechanism. 
We also give a more detailed characterization of the optimal prices for the case where there are only two possible demands. 		
		 
\end{abstract}


\section{Introduction}
%

We study a pricing problem that is motivated by the following examples. 
A cloud computing platform such as Amazon EC2 sells virtual machines to clients, each of who needs a different number of virtual machine hours. 
Similarly, cloud storage providers such as Dropbox have customers that require different amounts of storage. 
Software companies such as Microsoft sell software subscriptions that can have different levels of service. 
The levels could be the number of different documents you are allowed to create, or the number of hours you are allowed to use the software. 
Companies like Google and Microsoft sell API calls to artificial intelligence software such as face recognition, to other software developers. 
Video and mobile games are increasingly designed in such a way that one can pay for better access to certain features. 
Spotify and iTunes sell music subscription, and different people listen to  different number of songs in a month.
Cellphone service providers like AT\&T and Verizon offer cellular phone call minutes and data. People have widely varying amounts of data consumption. 
Dating apps provide paid services where certain number of messages sent by a client can be ``promoted". 

Pricing is an important component in all these examples. The aim of this paper is to understand how to price such goods, for a monopolist seller who aims to maximize her revenue. The following common features (to a first degree of approximation) of these examples are crucial for our model. 
\begin{itemize}
	\item The marginal cost of offering a higher level of service is essentially a constant (and in many cases zero). Most of the cost is a fixed cost. 
	\item The valuations for the different levels of service are roughly linear, subject to a cap. 
\end{itemize}
Based on this, we consider the following problem. There is a single good with multiple units of it for sale. Equivalently, there is a single service with various levels of service. 
For ease of presentation we simply refer to `goods' and `units' from now on. 
The marginal cost to the seller for procuring another unit of the good is a constant. There is a population of buyers, each of who has a {\em linear} valuation for 
consuming a number of units of the good, subject to a \emph{cap} (which we refer to as \emph{demand} henceforth). 
The type/private information of a buyer is determined by the per-unit valuation and the demand.  
Buyers try to maximize their utility which is quasi linear, i.e., the valuation minus their payment. 
The question is what is the revenue maximizing pricing scheme.  

The standard approach in  mechanism design \cite{Myerson} is Bayesian: 
assume that the types are drawn from a given distribution and find/characterize the incentive compatible (IC) mechanism
that maximizes the expected revenue when the buyer types are drawn from this distribution. 
The optimization is over randomized mechanisms, which in our case corresponds to pricing lotteries. 
Here's a simple example that shows that lotteries can obtain better revenue than any deterministic pricing scheme. 
We represent the type of a buyer with a pair $(v,d)$, where $v$ is the per unit valuation and $d$ is the demand. 
\begin{example}  [Deterministic pricing is not optimal]
	\label{eg:randopt}
	Suppose that there are 3 types of buyers, all occurring with equal probability $\frac{1}{3}$. 
These types are $t_1 = \left( 1 , 3 \right)$, $t_2 = \left( 1 , 2 \right)$ and  $t_3 = \left( 6 , 1 \right)$.
Consider the  lottery (which happens to be optimal for this case) that offers the following options: 
\begin{enumerate}
	\item  3 units at a price of 3, or
	\item a lottery that gives $2$ units at a price of 2 with probability $\frac{3}{4}$, and nothing otherwise. 
\end{enumerate}   
Buyers of type $t_1$ and $t_3$ buy 3 units, where as buyers of type $t_2$ buy the lottery, 
for a total expected revenue of $\frac{7.5}{3}$. 
Consider the two deterministic prices that are in the support of the lottery. 
The first one offers 3 units at a price of 3 and 2 units at a price of 2.
In this case, a type $t_3$ buyer will  switch to buying 2 units instead of 3 since her demand is only 1. 
Thus you get a revenue of $\frac{7}{3}$. 
The $1/4$ probability of not getting anything in the lottery makes the $t_3$ buyer pay for 3 units. 
The other price in the support of the lottery only offers 3 units at a price of 3. 
Buyers of type $t_1$ and $t_3$ buy this option whereas $t_2$ will not buy anything, 
resulting in a revenue of $\frac{6}{3} = 2$. 
It can in fact be argued that the revenue of $\frac{7}{3}$ is optimal for deterministic prices.\footnote 
{Any deterministic mechanism will have $3$, possibly distinct, prices for each possible number of units. Consider the price per unit for $2$ units $q_2$, and the price per unit for $3$ units $q_3$. If both $q_2$ and $q_3$ are strictly bigger than $1$ then $t_1$ and $t_2$ will not buy; the revenue in this case is at most $\frac{6}{3} = 2$. 
If $q_3 \leq 1$ and $q_2 > 1$, then $t_2$ will not buy (the price for one unit is larger than $q_2$, so buying one unit is not an option); in this case the maximum price we could charge for one unit is $3$, and therefore the maximum revenue attainable is $\frac{6}{3} = 2$. If $q_3 \leq 1$ and $q_2 \leq 1$ setting them equal to $1$ only increases revenue, for a maximum of $\frac{7}{3}$. The last case ($q_3 > 1, q_2 < 1$) is infeasible.
}
\end{example} 

It appears that the optimal mechanism is usually randomized for small examples with discrete support. 
This phenomenon is quite common. 
While \citet{Myerson} showed that the optimal mechanism for single dimensional types is deterministic 
under quite a general assumption about the prior distribution called regularity, 
even slight multi dimensional generalizations end up in randomized mechanisms as optimal \cite{thanassoulis2004haggling,pavlov2011optimal,hart2015maximal}. 
However, practical considerations force a seller to stick to deterministic mechanisms for the most part. 
(This is true for all the applications listed above.)
Moreover, the optimal randomized mechanism sometimes doesn't even have a finite description \cite{DaskalakisDT13}. 
Hence it is important to understand the structure of the optimal deterministic mechanism. 
In this paper we offer two insights in this regard. 

\textbf{Out first contribution is to identify a natural condition that guarantees that 
the optimal (randomized\footnote{We use the convention that the optimal mechanism is always randomized. When we wish to restrict ourself to deterministic mechanisms, we will use ``optimal deterministic mechanism/pricing''.}) mechanism is deterministic}. 
We call the condition we need as \emph{decreasing marginal revenue} (DMR), in accordance with previous literature \cite{che1998standard}.
Regularity requires that the \emph{virtual value function} is monotone, or equivalently, that the revenue function is  concave in the \emph{quantile} space.
DMR instead requires that \emph{in the value space}, the marginal revenue is decreasing or equivalently that revenue function is concave. In other words, a probability distribution with CDF $F$ is DMR, 
if the function $v(1-F(v))$, specifying the expected revenue of posting a price $v$, is concave. The condition we need for the optimal pricing to be deterministic is that the marginal distributions for $v$, 
conditioned on a given demand, are all DMR.  We will provide a more detailed analysis of the DMR condition below. We also give a detailed description of the optimal prices in case there are only two distinct demands in the distribution. 
We note that the case of 2 distinct levels of service is quite common (e.g., {\em limited} and {\em premium}). 

\paragraph{Closely Related Work}
\citet{malakhov2009optimal} consider the same problem (more generally in the multiple bidder case), 
but made 2 strong assumptions: 
(1) that the buyers cannot report a higher demand, and  
(2) that the distribution satisfies the following: the Myerson virtual value\footnote{
	The Myerson virtual value given a distribution with CDF $F$ and PDF $f$ is $\phi(v) = v - \frac{1-F(v)}{f(v)}$. 
	In our case, we define the virtual value of a type $(v,d)$ by applying the same definition using the  marginal distribution on  $v$, conditioned on $d$, and denoted by $F_d$ and $f_d$.
	$\phi(v,d) = v - \frac{1-F_d(v)}{f_d(v)}$. } is monotone in both the value and the demand. 
This essentially results in the problem separating out into a 1 dimensional problem, one for each $d$. 
The non-triviality in the 2 dimensional problem comes because buyers can misreport their demands. 
The first assumption disallows reporting a higher demand.  
The second assumption makes reporting a lower demand never profitable, without having to do anything extra.
When specialized to the case of a single buyer, it implies that a deterministic pricing is optimal, since the same is true for the 1 dimensional case. 

The recent work of \citet{fiat2016fedex} solves the single buyer problem, with only the first assumption above, that 
buyers cannot report a higher demand. 
This is a significant improvement, since the second assumption above, which requires something quite strong about the correlation between value and demand, is the more problematic one. 
\citet{fiat2016fedex} consider what they call the ``FedEx'' problem, 
which too has a 2 dimensional type space, where one of them is a value $v$, and the other  is a ``deadline'' $d$.  The seller offers a service, such as delivering a package, at various points of time, 
and the buyer's valuation is $v$ for any time that is earlier than her deadline $d$. 
In their model, a higher $d$ corresponds to an inferior product, as opposed to our model where higher $d$ is superior. 
The other difference is that in their model, all times earlier than $d$ have the same valuation and times later than $d$ have a zero valuation, whereas in our model, the valuation stays the same for higher $d$s but degrades gracefully as $d$ decreases. 

Despite these differences, the relevant IC constraints are \emph{syntactically} identical. 
As was also observed earlier by \citet{malakhov2009optimal}, without loss of generality, one can reduce the set of IC constraints under consideration to only ``local'' constraints, 
such as the ones where a buyer of type $(v,d)$ reports  $(v,d-1)$.
This IC constraint is exactly the same for both our problem and the FedEx problem. 
This is surprising because, as we observed above, what $d-1$ means in both cases is semantically different. (See Section \ref{sec:structure} for an explanation.)
On the other hand, a buyer with deadline $d$ can be made to never report a $d'>d$, by making sure that she is always given the service at her reported deadline. 
Thus, the FedEx problem is the same as our problem, with the assumption that the buyers are not allowed to report a higher demand. 
\citet{fiat2016fedex} characterize the optimal mechanism, without any assumptions on the prior distribution. 

\paragraph{Comparison.} We do not make the assumption that the buyers cannot report higher demands. 
Consider the case that there are just 2 different $d$s in the distribution, with $d_1 < d_2$, 
and the question, when is it optimal to offer each level of service 
at the monopoly reserve price (say, $r_1$ and $r_2$ resp.) for the corresponding marginal distributions over values. 
The answer for the FedEx problem is, when $r_1 \geq r_2$, which just says that $d_1$ should cost more than $d_2$. 
In our case, the answer is that $r_1 \leq r_2$ and $r_1 \geq \frac{d_1}{d_2} r_2$. 
Clearly $d_1$ units should cost less, but not too low either, since in that case some buyers with demand $d_2$ 
will actually prefer $d_1$ units.
This points to the added difficulty in our problem: we need to worry about a buyer opting for a bundle that could be of any size, but in the FedEx problem a buyer would never consider later time slots.  
In addition, the new IC constraints we need to consider are of the form where $(v,d)$ reports $(\frac d {d+1} v, d+1)$. 
These are ``diagonal'' IC constraints, as compared to the ``vertical'' ones in the FedEx problem, 
where a buyer of type $(v,d)$ reports  $(v,d-1)$. 
These are harder to handle and the techniques used in the FedEx problem, such as constructing an optimal dual, seem difficult to extend to this case.

\paragraph{The DMR Condition}
We are not the first to make this assumption: \citet{che1998standard} made the exact same assumption for a very similar problem, of selling a single item to a single buyer with budget constraints, rather than demand or capacity constraints. The optimal mechanism there could still be randomized. 
\citet{fiat2016fedex}  too show that the DMR condition is more natural than the usual notion of regularity for their setting. 
In particular, they show that to derive the optimal mechanism, one needs to \emph{iron}\footnote{Ironing is a technique introduced by \cite{Myerson} where the virtual value function is transformed so that it becomes monotone. This corresponds to transforming the corresponding revenue function into a concave function.} in the value space, 
rather than the quantile space as usual. 
DMR is precisely when no ironing is needed in the value space. 
As a result they too obtain that the optimal mechanism is deterministic under DMR. 
The same assumption was also made by \citet{kleinberg2003value} in the context of \emph{dynamic pricing}; see Section \ref{sec:concavityintro} for more discussion on dynamic pricing.

A simple class of DMR distributions is Uniform$[a,b]$ for any non-negative reals $a$ and $b$.  More generally, any distribution with finite support and monotone non-decreasing probability density is DMR.\footnote{The second derivative of the revenue function is $-2f(v) - vf'(v)$, which is negative if $f'(v)\geq 0$.} 
Another standard class of demand distributions that satisfies DMR is a \emph{constant elasticity} distribution.\footnote{As the name suggests, the elasticity of demand for such a distribution is constant over the support.  Such distributions are commonly used in Industrial Organization since they can be easily estimated by measuring elasticity anywhere on the support \cite{wolfstetter1999topics,berry1995automobile}.}   (See Example \ref{ex:regularity comparisons} for the definition.) 
The DMR condition is different from the regularity condition of \citet{Myerson}, which requires that the function $\phi(v) = v - \frac{1-F(v)}{f(v)}$ is monotone non-decreasing in $v$.  The example below shows that DMR and regularity are incomparable conditions.

\begin{example}[DMR vs. regularity]\label{ex:regularity comparisons}
Consider the class of \emph{constant elasticity} distributions with cumulative density $F(v) = 1 - (v/a)^{1/\epsilon}$ for any for $a\geq 0$ and $\epsilon < 0$, supported on $[a,\infty)$. A special case is when $a = 1$ and $\epsilon = -1$, in which case $F(v) = 1-1/v$, known as the \emph{equal revenue distribution}.  The corresponding revenue function $v(v/a)^{1/\epsilon}$ is concave if $\epsilon \leq -1$.  However, the function $\phi(v) = v - \frac{1-F(v)}{f(v)}$ simplifies to $v (1+\epsilon)$, which is monotone \emph{decreasing} for $\epsilon <1$.	 Therefore such a distribution is DMR, but not regular, for $\epsilon < -1$.  On the other hand, the exponential distribution is regular but not DMR.  Calculations for this example are straightforward and deferred to Appendix~\ref{app:appendix}.
\end{example}

The class of DMR distributions is well-behaved in the sense that it is closed under convex combinations.   In particular, the distribution that results from drawing a sample from a DMR distribution with probability $\alpha$, and from another DMR distribution	 with probability $1-\alpha$, is a DMR distribution.\footnote{The cumulative density of a distribution that samples from $F_1$ with probability $\alpha$, and from $F_2$ otherwise, is $F(v) = \alpha F_1(v) + (1-\alpha) F_2(v)$.  Therefore, the revenue function of the convex combination is the convex combinations of the revenue functions of $F_1$ and $F_2$, and is concave if $F_1$ and $F_2$ are DMR.}  On the other hand, it is known that regular distributions are not closed under convex combinations \cite{sivan2013vickrey}.

We show that the DMR condition is necessary, by giving a distribution with monotone hazard rate\footnote{The function $\frac{1-F(v)}{f(v)}$ is monotone non-increasing}, a condition stronger than regularity, for which a deterministic pricing is not optimal. 
\begin{example} [MHR distributions where deterministic pricing is not optimal]
The marginal distributions of Example~\ref{eg:randopt} for $d=1,2$ and $3$ are point masses at $6$, $1$ and $1$ respectively. Replace them with normal distributions $\mathcal{N}(1-\epsilon,\sigma)$, $\mathcal{N}(1-\epsilon,\sigma)$ and $\mathcal{N}(6-\epsilon,\sigma)$, truncated at $0$ and $V$, for some $V > 6$, and some $\epsilon > 0$. Truncated normal distributions satisfy the monotone hazard rate condition. For any $\delta > 0$, we can choose $\sigma$ and $\epsilon$ small enough, such that the revenue of the optimal deterministic and randomized mechanisms from Example~\ref{eg:randopt} changes by less than $\delta$. Furthermore, running these mechanisms on the new distributions yields essentially the same revenue.
\end{example}

\paragraph{Our Approach.} Our approach is to show that any mechanism can be converted to a deterministic one with higher revenue, which we perform in two steps.  First, we convert a mechanism so that a type with demand $d$ and with highest valuation receives a deterministic allocation of $d$ units, without reducing revenue.  In order to do so, we first argue that without loss of generality, any type $(v,d)$ is assigned a lottery over $d$ units or no allocation (that is, there is no chance of receiving $d' \neq d$ units).  Then we show that the randomized allocation of highest values can be converted to a deterministic allocation, without reducing revenue. Our first step holds generally and does not require the DMR condition. Second, we argue that a mechanism resulting from the first step can be converted to a deterministic mechanism.  In particular, we remove all non-deterministic allocations from the mechanism, and allow types to choose only among the remaining deterministic allocations.  Removal of allocations can only decrease (or keep fixed) the utility function of the mechanism pointwise.  However, since the highest type of each demand was assigned a deterministic allocation, the utility of such a type remains unchanged.  A technical lemma shows that under the DMR condition, we can improve revenue by pointwise lowering utility whilst fixing the utility of highest types.

\subsection{Concavity of the revenue function}\label{sec:concavityintro} 
Our first result implies that the optimal pricing scheme is a price vector, which offers each number of units for a given price. 
\textbf{Our second contribution is to show that the revenue as a function of the price vector is concave}, under the 
same assumption of DMR. 
This implies that the optimal prices can be found efficiently 
using the ellipsoid or other cutting plane methods \cite{khachiyan1980polynomial,vaidya1989new,lee2015faster}. 
Note that DMR is a condition on the marginal distributions of values, and does not immediately imply concavity as a function of the vector of prices.  Note also that when we define concavity, we consider a deterministic pricing scheme where the price vector is a convex combination of two other deterministic prices, and not the corresponding lottery. This is best illustrated with the same instance as in Example \ref{eg:randopt}, which also shows that the revenue function need not always be concave. 
\begin{example}[Revenue function is not concave]
	Consider the instance from Example \ref{eg:randopt}, and the convex combination of the two prices in the support of the lottery, using the same convex combination of $(3/4, 1/4)$ as before. Recall that the first price vector is 3 units at price 3 and 2 units at price 2 for a revenue of 7, and the second is a price of 3 for either 2 or 3 units, for a revenue of 6. 
	The convex combination offers 3 units at a price of 3, and 2 units at a price of $9/4$, 
	with a revenue of 6. 
	The corresponding convex combination of the revenues is strictly larger than 6, and hence the revenue function is not concave. 
	\end{example}	

\paragraph{Techniques and Difficulties}
In order to show that the revenue function is concave, we first give a closed form formula for the revenue function \emph{region-wise}. 
We divide the price space into different regions such that a region determines the order in which a buyer with a certain demand actually 
ends up buying a lower sized bundle. For instance, a region might determine that for all the buyers with demand 10, as their value 
decreases from $\infty$ down to 0, the bundle size they actually buy goes from 10 to 7 to 3 to 0; 
the exact transition points of course depend on the prices. 
We then show that the closed form formula for each of the regions is a concave function, implying that the revenue is piecewise concave. 
This in general does not imply that the revenue function is concave everywhere. 
One might surmise that the revenue function is the minimum of each of these functions, which would show that it is concave everywhere, 
but that is unfortunately not true. 
In fact, there is a partial order over these functions such that some of them are always higher than the others. 
We show a somewhat surprising property, that at the boundaries of the regions where they intersect, not only do the different functions 
agree (which they should, for the revenue function to be even continuous), but also their gradients agree! 
Showing this involves arguing that the equalities that hold at an intersection imply a whole set of other equalities 
such that disparate terms in the two gradients cancel out.

\paragraph{Dynamic Pricing}
As a corollary, we obtain that under the DMR assumption, there is an efficient \emph{dynamic pricing} scheme, defined as follows.
Consider a repeated setting where in each round  $\tau \in \{1,2,\ldots, T\},$ 
the seller posts a price vector $\pb^{\tau}$, a buyer is drawn from a fixed distribution, and buys her utility maximizing bundle. 
The seller does not know the distribution of buyer types, and has to only use the purchase information in previous rounds to set the price. 
The goal is to approach the optimal revenue as $T$ goes to infinity. 
Given that the distribution satisfies the DMR assumption, our result on the concavity of the revenue curve implies that this is a special case of the ``convex bandits'' problem \cite{agarwal2011stochastic,bubeck2016kernel}. 
The results of \citet{bubeck2016kernel} imply that there exists a dynamic pricing scheme such that the average revenue per round 
converges to the optimal revenue at the rate of 
$ \frac{n^{9.5}}{\sqrt{T}}$, where $n$ is the number of units.
These bounds are quite strong, since the best known bounds for the dynamic pricing problem in general scale 
\emph{exponentially} in $n$; 
the concavity of the revenue function is an assumption often made to escape this curse of dimensionality \cite{besbes2012blind,talluri2006theory}. 
We show that this assumption can be weakened to an assumption about the concavity of only the 1 dimensional revenue functions for each $d$. 
The same assumption was made by \citet{kleinberg2003value} to get a $1/\sqrt{T}$ regret for the case of a \emph{single} item. 

\subsection{Other Related work}
%

The seminal work of \citet{Myerson} settled the optimal mechanism design problem for selling to multiple buyers with single parameter type spaces. 
Since then, it has been discovered that multi-dimensional type spaces are a lot more difficult to analyze, and this remains to this day the 
foremost challenge in mechanism design. 
The optimal mechanism becomes randomized for even slight generalizations \cite{thanassoulis2004haggling,pavlov2011optimal,hart2015maximal}. 
Following \citet{Myerson}, some early work solved very special cases of this \cite{laffont1987optimal,mcafee1988multidimensional}.  
\citet{manelli2006bundling} showed conditions under which bundling all the items was optimal when there were either 2 or 3 heterogeneous items. 
Success with reasonably general settings had been limited. 

There has been a recent spate of results in the algorithmic game theory community characterizing optimal mechanisms for special cases, and all of these consider a single buyer. 
\citet{DaskalakisDT13} use optimal transport theory to give sufficient conditions of optimality for an additive buyer with independent item valuations: when ``selling only the grand bundle'' is optimal and examples where a continuum of lotteries is the unique optimal mechanism.
\citet{GiannakopoulosK14} identify a (deterministic) optimal auction for an additive buyer whose valuations are i.i.d. from $U[0,1]$, 
for up to 6 items. 
\citet{HaghpanahH14} identify conditions under which either ``selling only the favorite item''  for a unit-demand buyer or selling only the grand bundle for an additive buyer is optimal. 
\citet{DaskalakisDT15} identify necessary \emph{and} sufficient conditions for selling only the grand bundle to be optimal for an additive buyer. 
The FedEx problem \cite{fiat2016fedex} that we described earlier also falls in this line of work. 
Our paper contributes to this line of work by identifying a reasonably general setting where the optimal mechanism is in fact deterministic, and can be computed efficiently. 
All these results use linear or convex program duality, to  construct a witness (dual optimal solution) of optimality. 
We also frame our problem as a mathematical program, but argue about the primal directly, which we find gives more intuition.\footnote{We did try to construct the optimal duals explicitly, but were not able to construct such duals in general. Constructing such duals is likely to facilitate characterizing the optimal mechanism for all distributions.}

The lack of characterizations of optimal mechanisms in general settings has been addressed by 
seeking computational results instead.  
(We refer the reader to \citet{hartline2013mechanism} for a thorough overview of this line of work.)
A sequence of papers by \citet{CaiDW12a,CaiDW12b,CaiDW13a,CaiDW13b} showed that for finite (multi-dimensional) type spaces, 
the mechanism design problem can be reduced to a related algorithm design problem, thus essentially resolving the computational question for this case. 
Most of these assume a finite support and the computation time is polynomial in the size of the support. 
This is different from our model which assumes a continuous distribution. 

Yet another approach to cope with the complexity of optimal mechanisms has been to show that simple  auctions approximate optimal ones. 
In this line of work, two classes of valuations have been widely studied, unit demand valuations \cite{ChawlaHK07,BriestCKW10,ChawlaHMS10,ChawlaMS10,Alaei11}, 
and additive valuations \cite{HartN12,li2013revenue,BabaioffILW14,Yao15}. A unified approach to both has been presented in \citet{cai2016duality}, 
and these approaches have been extended to more general valuations in \citet{RubinsteinW15,chawla2016mechanism,CaiZ16}. 
Most of these make some sort of assumption about independence of values for different items. 
Our model differs in this aspect: either we see it as a special case of a unit demand problem (each buyer wants one of several bundles) 
in which case the values are highly correlated, or as a problem with 2 dimensional type space $(v,d)$, and we allow arbitrary correlations between the $v$ and $d$. 
Also, the goal in our paper  is a characterization of the optimal mechanism as opposed to identifying simple but approximately optimal mechanisms.


%

\section{Model and Main Results}\label{sec:model} 
We consider a {\em multi-unit mechanism} with a single buyer with {\em private} demand. 
In a multi-unit mechanism, there are infinitely many units of a single item for sale. 
The type $t$ of a buyer is specified by her per unit value $v\in \reals_+$ and her demand $d\in \integers_+$. 
The valuation of such a buyer for $m\in \integers_+$ units of the item is $v * \min \left\{ m , d \right\}$. 
Both $v$ and $d$ are private information of the buyer, making this a multi-parameter setting. 

We restrict our attention to direct revelation mechanisms, which ask the buyer to report her type $t=(v,d)$. 
The mechanism is allowed to be randomized, so the output is an allocation $A \in \integers_+$ and a payment $P\in \reals_+$, both of which are random variables (and functions of the reported type $(v,d)$). 

We require the mechanism to be incentive compatible, in expectation over the randomization of the mechanism. Formally, a mechanism is said to be EIC if for all valid types $(v,d)$ and $(v',d')$,  the utility of the type $(v,d)$ from reporting its type truthfully is at least the utility it would get from reporting type $(v',d')$,
\[ \textstyle \E \left[   v\left(\min \left\{ A(v,d) , d \right\} - \min \left\{ A(v',d') , d \right\} \right) -  
P(v,d) + P(v',d') \right]  \geq 0 ,\]
where the expectation is taken over the randomization of the mechanism. 
We assume that $(0,0)$ is a always a valid type declaration, so this includes as a special case, an 
expected individual rationality (EIR) condition, which requires that each type must get a non-negative utility from reporting its type truthfully
\begin{align*}
\textstyle \E \left[   v \min \left\{ A(v,d) , d \right\} -  
P(v,d) \right]  \geq 0.
\end{align*}
\noindent By linearity of expectation, we may assume w.l.o.g. that the payment is deterministic, and we denote this deterministic payment by $p(t)$.   A stronger notion of individual rationality is \emph{ex-post} individual rationality, which requires that the utility of a type is positive for any randomization of the mechanism.  However, in the lemma below we show that any EIR mechanism can be converted to an \emph{ex-post} individually rational mechanism which guarantees positive for any randomization of the mechanism.  The argument is standard and is deferred.  As a result of the lemma, we will only focus on the EIR constraint in what follows. (All the missing proofs in the rest of the paper are in Appendix~\ref{app:appendix}.)

\begin{restatable}{lemma}{expost}
For every EIC and EIR mechanism, there exists an EIC and ex-post IR mechanism with the same expected payment \emph{for any type}.
\end{restatable}

When there are no supply constraints that bind across buyers, or equivalently there is a single buyer, an alternate interpretation of such a mechanism is as a \emph{menu of lotteries}. 
A lottery is a pair of a probability distribution over $\integers_+$ and a price, 
corresponding to the randomized allocation and payment. 
The buyer chooses the lottery that maximizes her expected utility from among a menu. 
In general this menu could be of infinite size. 
We call this the {\em multi-unit pricing} problem.

Consider a distribution over the type space, with a density function $f$. 
The \emph{Bayesian optimal} mechanism w.r.t. this distribution is the EIC (and EIR) mechanism that 
maximizes the expected revenue when the types are drawn from this distribution: 
\[ \textstyle \E_{t\sim f} \left[ p(t)\right]  . \]
Our goal is to characterize the Bayesian optimal mechanism. 
We make two assumptions: 
\begin{itemize}
	\item The support of the distribution in the demand dimension is finite.  We denote by $k$ the size of this support. In other words, there are $k$ different demands possible. 
	\item Let $f_d$ and $F_d$ denote the PDF and the CDF of the marginal distribution on values conditioned on the demand being $d$. Then $ v(1-F_d(v))$ is \textbf{concave} in $v$ for any given $d$. 
We call this property \emph{decreasing marginal revenue} (DMR). 
This is equivalent to the fact that $v f_d(v)  - {1-F_d(v)}$ is a non-decreasing function of $v$. 
This is closely related to the usual definition of regularity, which requires monotonicity of this function divided 
by $f(v)$. 
\end{itemize}

We now state our first main theorem.
\begin{restatable}{theorem}{deterministic}
\label{thm:deterministic}
	The Bayesian optimal multi-unit pricing with linear valuations, private demands, finitely many demands and DMR distributions is deterministic.
\end{restatable}
A deterministic mechanism is simply a menu with a deterministic allocation of each possible bundle of $d$ units, 
for $d$ in the support of $f$. Let $d_1 < d_2 < \cdots < d_k$ be the demands in this support. We denote 
the  prices for the corresponding bundles by  $p_1 , p_2, \ldots, p_k$. 
A buyer can get $d_i$ units by paying $p_i $ for any $i\in [k]$. 
A buyer with type $t = (v,d)$ chooses to buy the bundle that maximizes her utility
$ v \min\{ d, d_i\} - p_i.$
Let $\pb $ denote the vector of unit prices $(p_1, \ldots, p_k)$. 
We assume without loss of generality that the domain of $\pb$ is such that 
$p_1 \leq p_2 \leq \cdots \leq p_k$. 
We denote by $\Rev(\pb)$ the (expected) revenue of this mechanism.  
Our second main theorem is Theorem~\ref{thm:concavity}. Due to this theorem, the optimal mechanism can be found efficiently, since maximizing a concave function can be done in polynomial time. 
\begin{restatable}{theorem}{concavity}\label{thm:concavity}
	$\Rev(\pb)$ is a concave function if the marginal distributions are DMR for all $d$. 
\end{restatable}

\paragraph{Dynamic pricing}
Consider the following online problem. In each round $\tau \in {1,2,\ldots, T}$, for some $T \in \integers_+$, 
the following takes place. 
\begin{enumerate}
	\item The seller posts a price vector $\pb^{\tau}$. 
	\item A buyer of type $(v^\tau, d^\tau)$ is drawn independently from the distribution $f$. 
	\item The buyer buys her utility maximizing bundle $x^\tau \in \arg\max_{\{i: d_i \leq d^\tau\}} v^\tau d_i - p^\tau_i   $. 
	\item The seller observes only $x^\tau$. 
\end{enumerate}
Assume, for the sake of notational convenience, that $d_0 = 0 $ and $p^\tau_0 = 0$ for all $\tau$, 
so  $x^\tau = 0$ when the buyer doesn't buy anything.
The goal of the seller is to maximize her average (or equivalently, total) revenue
\[  \frac 1 T \sum_{\tau=1}^{T}  p_{x^{\tau}}^\tau.  \]
We evaluate the performance of  a dynamic pricing scheme by its \emph{regret}, which is 
the difference between the optimal expected revenue and the average expected revenue of the pricing scheme. 
We assume that the values are bounded, and that $v_{\max}$ is the maximum value. 
The results of \citet{bubeck2016kernel,bubeck2017personal} imply the following as a corollary of Theorem \ref{thm:concavity}. 
\begin{corollary} 
There is a dynamic pricing scheme where the regret  is 
\[ \frac{\tilde{O}(n^{9.5}) d_k v_{\max}} {\sqrt{T}}.\]
\end{corollary}

\section{Deterministic mechanisms are optimal}
\label{sec:structure}
In this section we prove our first main theorem. Throughout this section, we assume, for the sake of convenience, that the support of the distribution in the value space is $\subseteq [0,\bar{V}]$. 

\deterministic*

\paragraph{Allocating only the demanded:} 
We first use a reduction that might actually introduce randomization: w.l.o.g. we may assume that $A(v,d)$ is supported on $\{0,d\}$. A buyer who reports a
demand of $d$ is either allocated exactly $d$ units or none at all. 
The reduction replaces any allocation of $d' <d $ units with an allocation of $d$ units with probability $d'/d$ while retaining the same payment, and argues that this does not violate any EIC constraints.   
This may seem to go counter to our eventual conclusion that deterministic pricing is optimal; 
there are easy examples where a deterministic optimal pricing allocates  $d' < d$ units. 
Nonetheless, what we will show in the end is that the allocation probabilities for a buyer with demand $d$ should be exactly equal to  
$d'/d $ for some other (lower) demand $d'$. 
We can then reduce in the other direction: this is equivalent to deterministically allocating $d'$ units.

\begin{restatable}{lemma}{Support}\label{lem:support}
For every feasible Bayesian mechanism, there exists another mechanism, with revenue at least as large, such that $A(v,d)$ is supported on $\{0,d\}$. 
\end{restatable}

Let $t_i = (v_i,d_i)$ and $t_j = (v_j,d_j)$ be any two types. 
We  write $\ut{v_i,d_i}{v_j,d_j}$, or just $\ut{t_i}{t_j}$, for the utility of an agent with type $t_i$ 
when she reports type $t_j$.
From now on, we assume that the mechanism allocates $d_i$ units to $t_i$, with some probability $w\left(t_i\right)$, and for some price $p(t_i)$.
Using this, $\ut{t_i}{t_j}$ can be re written as $ v_i \min\left( d_i , d_j \right) w\left(t_j\right) - p(t_j)$. We write $w_d$ for the allocation probability as a function of $v$ when the reported demand is $d$.

\paragraph{Local IC constraints are sufficient:}
We now show that it is sufficient to consider a subset of IC constraints; the others are implied by these. 
The first set of constraints are ``horizontal'' constraints, where you fix $d$ and only change $v$. 
Further, the horizontal constraints can be replaced by monotonicity and a payment identity \`a la Myerson: 
\[ \textstyle
\p{v,d} = vd w_{d}(v) - d \int_0^v  \w[d]{z} dz + \p{0,d}.
\]
\noindent We now argue that in the optimal mechanism we must have $\p{0,d} = 0$ for all $d$.  Incentive compatibility requires that $\p{0,d} = \p{0,d'}$ for all $d,d'$, since otherwise the type with higher payment would prefer to report being the other type and pay less (such a type gets no utility from allocation).  The next step is to show that an mechanism where $\p{0,\cdot}<0$ cannot be optimal.  To see this, construct another mechanism which adds $\p{0,\cdot}$ to the payment of \emph{all types}.  The new mechanism respects all the EIC and EIR constraints (utility of type $(0,d)$ is zero for all $d$), and has higher revenue.  As a result, the payments identity simplifies to: 
\begin{equation}\label{eq:paymentid} 
\textstyle \p{v,d} = vd w_{d}(v) - d \int_0^v  \w[d]{z} dz.
\end{equation}
In addition to the local horizontal constraints consider above,  there are the local ``vertical''  constraints, which are of two types; a type with demand $d_i$ reports $d_{i +1}$ or $d_{i-1}$.
In either case, we only need to consider a particular misreport of the value $v'$, and this value is 
such that $\ut{v,d}{v',d'} = \ut{v',d'}{v',d'}$. The following lemma characterizes such $v'$, which can be verified by an easy calculation. 

\begin{restatable}{lemma}{utilandX}\label{lem:utilandX}
	$\ut{v,d_i}{v\frac{d_i}{d_j},d_j} = \ut{v\frac{d_i}{d_j},d_j}{v\frac{d_i}{d_j},d_j}$ for $j>i$,  and \\
	$\ut{v,d_i}{v,d_j} = \ut{v,d_j}{v,d_j}$
	for $j<i$. 
\end{restatable}

The next lemma formalizes our discussion above on sufficiency of local EIC constraints.  The first condition of the lemma is the local horizontal constraint, and the next two are local vertical  constraints.  The lemma follows by showing that the EIC constraint where $(v,d)$ misreports $(v',d')$ is implied by
a sequence of EIC constraints, where you iteratively use the vertical constraints to change the report of $d$ by $\pm1$ until you get to $d'$, and then use the horizontal constraint to change the report to $v'$. 
\vspace{-1mm}
\begin{restatable}{theorem}{localtoglobal}
\label{thm:local_to_global}
A mechanism satisfying the following conditions is EIC: $\forall d_i$ and $\forall v$,
\begin{enumerate}
\item $w_{d_i}$ is monotone non-decreasing, and $p(v,d_i)$ is given by Equation~\eqref{eq:paymentid}. 
\item $\ut{v,d_{i+1}}{v,d_{i+1}} \geq \ut{v,d_{i+1}}{v,d_i}$
\item $\ut{v,d_i}{v,d_i} \geq \ut{v,d_i}{v \frac{d_i}{d_{i+1}}, d_{i+1}}$
\end{enumerate}
\end{restatable}\vspace{-2mm}
It is interesting to compare this  global-to-local reduction with that used in the FedEx problem. 
Syntactically, for the the FedEx problem just the first 2 constraints above are sufficient, 
but the semantics are different. 
In the FedEx problem the $d$'s are the deadlines, and a larger $d$ signifies an inferior product, 
whereas in our problem a larger $d$ is a superior product. 
That the EIC constraints still look the same for misreporting a lower $d$ is due to the other difference between the problems: 
 utility scales linearly with $d$ in our problem, but remains constant in the FedEx problem. 
Thus in both problems, the valuation for an item of type $d'<d$ is the same for types $(v,d)$ and $(v,d')$.

\vspace{-1mm}
\paragraph{Mathematical Program for the optimal mechanism:}

We now write a mathematical program that captures the optimal mechanism. 
It will turn out to be convenient to use the following as variables of the program. 
Let $U_{d_i}(v) := \int_0^v \w[d_i]{z} dz$. 
Notice that $d_i U_{d_i}(v)$ is just the utility of a type $(v,d_i)$ when reporting the truth. 
Our objective is to maximize revenue, i.e. $\sum_{d_i = 1}^D \int_0^{\bar{V}} \p{v,d_i} f(v,d_i) dv$. 
Let $\phi_{d}(v) := v - \frac{1 - F_{d}(v)}{f_{d}(v)}$ be the standard Myerson virtual value function.
Using the payment identity (\ref{eq:paymentid}) and integration by parts \`a la Myerson, we can rewrite this objective 
in terms of the $U_{d_i}(v) $ variables as:
\vspace{-1mm}
\begin{align}
\textstyle Rev &= \sum_{d_i = 1}^{\total} \int_0^{\bar{V}} \w[d_i]{v} \phi_{d_i}(v) f_{d_i}(v) dv = \sum_{d_i = 1}^{\total} \int_0^{\bar{V}} U'_{d_i}(v) \phi_{d_i}(v) f_{d_i}(v) dv \nonumber \\
 &= \sum_{d_i = 1}^{\total} U_{d_i}\left( \bar{V} \right) \phi_{d_i}(\bar{V}) f_{d_i}(\bar{V}) - \int_0^{\bar{V}} U_{d_i}(v) \left( \phi_{d_i}(v) f_{d_i}(v) \right)' dv.\label{re:revenue utility}
\end{align}

Using this, and Theorem~\ref{thm:local_to_global}, we can restate the Bayesian optimal mechanism design problem as the following program. 
We define $U'_{d}$ to be the left derivative of $U_{d}$, which will be convenient to think of as $w_{d}$, the probability of allocation. 
Note that since the distribution over types is continuous, whether we allocate or not to any particular type $(v,d)$ does not affect revenue.
The first constraint is equivalent to saying that the allocation is monotone non decreasing, 
and the second constraint says that the allocation probability is between $0$ and $1$. 

\vspace{-3mm}
\begin{talign}
\textrm{max}    & \sum_{i=1}^k  U_{d_i}\left( \bar{V} \right) \phi_{d_i}(\bar{V}) f_{d_i}(\bar{V}) - \int_0^{\bar{V}} U_{d_i}(v) \left( \phi_{d_i}(v) f_{d_i}(v) \right)' dv & \label{eq:mathprog} \notag \\
\textrm{subject to}  &:\notag\\
& U_{d_i}(v) \text{ is concave}& \forall i \in [k] \notag \\
& 1 \geq U'_{d_i}(v) \geq 0 & \forall i \in [k], \forall v\\
& U_{d_i}(0) = 0 & \forall i \in [k] \notag\\
& d_i U_{d_i}(v) \geq d_{i+1} U_{d_{i+1}}\left( v \frac{d_i}{d_{i+1}}\right) & \forall i \in [k-1] \notag\\
& d_{i+1} U_{d_{i+1}}(v) \geq d_i U_{d_i}\left( v \right) &  \forall i \in [k-1]	\notag				
\end{talign}

To prove our main result, we utilize the the DMR property through the following Lemma. The Lemma allows us to compare the revenue of mechanisms by pointwise comparing their induced utility functions.  In particular, the Lemma states that by lowering the utilities of all types while keeping the utility of types with the highest value fixed, we can improve the revenue of a mechanism.

\begin{lemma}\label{lem:reveneu and utility}
	Consider two feasible mechanisms with utility functions $U$ and $\bar{U}$, such that $U_d(v) \leq \bar{U}_d(v)$ for all types, and $U_d(\bar{V}) = \bar{U}_d(\bar{V})$ for all $d$. If the marginal distributions $F_d$ are DMR for all $d$, then the revenue of the mechanism with utility function $\bar{U}$ is at least as high as the revenue of the mechanism with utility function $U$.
\end{lemma}
\begin{proof}
	The proof follows directly from the expression of revenue in equation \eqref{re:revenue utility}.  Since $U_d(\bar{V}) = \bar{U}_d(\bar{V})$, for all $i$ we have 
	\begin{align}
	U_{d_i}\left( \bar{V} \right) \phi_{d_i}(\bar{V}) f_{d_i}(\bar{V}) &= \bar{U}_{d_i}\left( \bar{V} \right) \phi_{d_i}(\bar{V}) f_{d_i}(\bar{V}).\label{eq:rev part 1}
	\intertext{In addition, $\phi_d(v)f_d(v) = vf_d(v) - (1-F_d(v)) = -\frac{d}{dv}(v(1-F_d(v)))$.  The assumption that $v(1-F_d(v))$ is concave implies that $\frac{d}{dv}(v(1-F_d(v)))$ is monotone non-increasing, which implies that $\phi_d(v)f_d(v)$ is monotone non-decreasing, or equivalently $(\phi_d(v)f_d(v))' \geq 0$.  The assumption that $U_d(v) \leq \bar{U}_d(v)$ then implies that}
	- \int_0^{\bar{V}} U_{d_i}(v) \left( \phi_{d_i}(v) f_{d_i}(v) \right)' dv &\leq - \int_0^{\bar{V}} \bar{U}_{d_i}(v) \left( \phi_{d_i}(v) f_{d_i}(v) \right)' dv.\label{eq:rev part 2}
	\end{align}
	\noindent By equation~\eqref{re:revenue utility}, the revenue of the mechanism with utility function $U$ is
\[ \textstyle \sum_{d_i = 1}^{\total} U_{d_i}\left( \bar{V} \right) \phi_{d_i}(\bar{V}) f_{d_i}(\bar{V}) - \int_0^{\bar{V}} U_{d_i}(v) \left( \phi_{d_i}(v) f_{d_i}(v) \right)' d \]
\[ \textstyle \leq \sum_{d_i = 1}^{\total} \bar{U}_{d_i}\left( \bar{V} \right) \phi_{d_i}(\bar{V}) f_{d_i}(\bar{V}) - \int_0^{\bar{V}} \bar{U}_{d_i}(v) \left( \phi_{d_i}(v) f_{d_i}(v) \right)' dv, \]
	\noindent which is equal to the revenue of the mechanism with utility function $\bar{U}$.	
\end{proof}

Having set the problem up, we now turn to the proof of the main theorem, \autoref{thm:deterministic}, that a deterministic mechanism is the optimal solution to the above revenue maximization program.  The main component is the Lemma below, which shows that for any feasible solution to the problem, there exists a deterministic solution with revenue at least as large.

\begin{lemma}{\label{lem:deterministic}}
	Consider any feasible solution to mathematical program   (\ref{eq:mathprog}). If the marginal distributions $F_d$ are DMR for all $d$, then there exists a deterministic mechanism with revenue at least as large.
\end{lemma}

The Lemma follows immediately from the following two Lemmas.

The first Lemma states that we can improve the revenue of any mechanism by assigning any type $(\bar{V},d)$ a deterministic allocation of $d$ units.  In particular, we show that by setting the price for the deterministic allocation of $d$ units appropriately, we can ensure that the type $(\bar{V},d)$ would be willing to choose the deterministic allocation, while no other type would have the incentive to misreport and get that allocation. The intuition is that a type $(\bar{V},d)$ has, among all other types, the highest value per unit for a deterministic allocation of $d$ units.  By setting the price of the allocation in a way that $(\bar{V},d)$ is indifferent, no other type would be willing to take the new allocation.  In addition, since $(\bar{V},d)$ is indifferent between the deterministic allocation of $d$ units and its previous allocation, and since it has higher value for $d$ units, its payment for $d$ units has only increased. Note that the Lemma below does not require the DMR condition.

\begin{lemma}{\label{lem:deterministic at top}}
	Consider any feasible solution to the mathematical program (\ref{eq:mathprog}).  There exists a mechanism, with revenue at least as large, where any type with highest value $(\bar{V},d)$ deterministically receives $d$ units.
\end{lemma}
\begin{proof}
	Fix any feasible mechanism $(w,p)$.  Construct a mechanism $(\bar{w},\bar{p})$ as follows.  For each demand $d$, define $\bar{w}_{d}(\bar{V})=1$ and $\bar{p}_{d}(\bar{V}) = \bar{V}d - U_{d}(\bar{V})$.  All other types $(v,d)$ with $v < \bar{V}$ are assigned the same allocation and payment as in the original mechanism.  
	
	We first argue that any type obtains the same utility from reporting truthfully in both mechanisms, that is 
	\begin{equation}\label{eq:equal utility}
	\textstyle U_{d}(v) = \bar{U}_{d}(v). 
	\end{equation}
	For any type $(v,d)$ where $v< \bar{V}$, the allocation and the payment remains the same.  Any type $(\bar{V},d)$ satisfies $\bar{U}_{d}(\bar{V}) = \bar{V}d -  (\bar{V}d - U_{d}(\bar{V})) = U_d(\bar{V})$.  Since the original mechanism is individually rational, so will be the new mechanism given $U_{d}(v) = \bar{U}_{d}(v)$.  
	
	Further, notice that the revenue of the mechanism $(\bar{w},\bar{p})$ is no lower than the revenue of $(w,p)$.  In fact, we have $\bar{p}_d(\bar{V}) = \bar{V}d - U_{d}(\bar{V}) = \bar{V}d - (\bar{V}dw_{d}(\bar{V}) - p_{d}(\bar{V})) \geq p_{d}(\bar{V})$, while payments of all other types remain the same.
	
	We next argue that the mechanism $(\bar{w},\bar{p})$ is incentive compatible.  We only need to show that a type $(v,d)$ has no incentive to misreport to $(\bar{V},d')$.  The utility from misreporting is
	\begin{align*}
	\bar{u}(v,d \rightarrow \bar{V},d') &= v\min(d,d') - \bar{p}_{d'}(\bar{V}) \\ 
	& = v\min(d,d') - (\bar{V}d' - U_{d'}(\bar{V})) \\
	&= v\min(d,d') - (\bar{V}d' - (\bar{V}d'w_{d'}(\bar{V}) - p_{d'}(\bar{V}))) \\
	& = v\min(d,d') - \bar{V}d'(1 - w_{d'}(\bar{V}))- p_{d'}(\bar{V}). \\
	\intertext{Since $1 - w_{d'}(\bar{V}) \geq 0$, we conclude that}
	\bar{u}(v,d \rightarrow \bar{V},d')& \leq v\min(d,d') - v\min(d,d')(1 - w_{d'}(\bar{V}))- p_{d'}(\bar{V}) \\
	&= v\min(d,d') w_{d'}(\bar{V})- p_{d'}(\bar{V}).
	\intertext{The above expression is the utility that type $(v,d)$ would obtain from misreporting type $(\bar{V},d')$ in the original mechanism.  By incentive compatibility of $(w,p)$, the above expression is at most $u(v,d \rightarrow v,d)$. Therefore, we conclude that}
    \bar{u}(v,d \rightarrow \bar{V},d')& \leq u(v,d \rightarrow v,d) = U_d(v) = \bar{U}_d(v),
	\end{align*}
	where the last equation is the same as equation~\eqref{eq:equal utility}, and was established above.  Thus the mechanism is incentive compatible, and the Lemma follows.
\end{proof}

The next Lemma builds on Lemma~\ref{lem:deterministic at top} and shows that for any mechanism where any type $(\bar{V},d)$ deterministically receives $d$ units, there exists a deterministic mechanism with revenue at least as large.  The intuition is that by removing all non-deterministic allocations from the mechanism, the utility of every type would weakly decrease, while the utility of a type $(\bar{V},d)$ stays the same.  Lemma~\ref{lem:reveneu and utility} can then be used to argue that the revenue of a deterministic mechanism is weakly higher.

\begin{lemma}
	Consider any mechanism where any type with highest value $(\bar{V},d)$ deterministically receives $d$ units.  If the marginal distributions $F_d$ are DMR for all $d$, then there exists a deterministic mechanism with revenue at least as large.
\end{lemma}
\begin{proof}
	Fix any type with highest value $(\bar{V},d)$ that deterministically receives $d$ units.   Consider the menu representation of the mechanism: it offers, among other lotteries, deterministic allocations of $d$ units, for all $d$.   Now construct an alternative menu that only offers such deterministic allocations. The alternative menu contains $\total$ choices of deterministic allocations of $d_1$ to $d_{\total}$ units.  Note that the utility function of the alternative mechanism is pointwise (weakly) smaller than the utility function of the original mechanism, since each type faces a smaller menu of choices.  Furthermore, the utility of type $(\bar{V},d)$ remains the same for all $d$, since the deterministic allocations that they chose in the original mechanism are still available in the alternative mechanism.  By Lemma~\ref{lem:reveneu and utility}, the revenue of the alternative mechanism is no lower than the revenue of the original mechanism.
\end{proof}

We are now ready to complete the proof of \autoref{thm:deterministic}.

\begin{proof}[Proof of \autoref{thm:deterministic}]
	Consider any feasible solution to the problem.  By Lemma~\ref{lem:deterministic}, the revenue of the mechanism is at most the revenue of the optimal deterministic mechanism.  Since the optimal deterministic mechanism exists and is a feasible to the problem, it must also be the optimal solution to the problem~\ref{eq:mathprog}.
\end{proof}

\section{Concavity of the revenue function} 
In this section we prove Theorem~\ref{thm:concavity}.  
Recall that the demands in the support of the distribution are $d_1 < d_2 < \cdots < d_k$, and 
that for all $i \in [k]$, $p_i$ denotes the  price for the bundle of $d_i$ units, and $\pb$ denotes the vector of all $p_i$s.  
Without loss of generality, we may assume that the domain of $\pb$ is 
\[ 0 \leq p_1  \leq p_2 \leq \cdots \leq p_k . \]
With this, we may assume that a buyer with demand $d_i$ only buys a bundle $d_j$ for $j\leq i$. 
We restate Theorem~\ref{thm:concavity} for convenience. 
\concavity*

\paragraph{Characterizing optimal bundles:}
The revenue is determined by what the optimal bundle for each type is, given a price $\pb$. 
To analyze this, we first consider when a given type prefers a bundle of $d_j$ units to one of $d_l$ units, 
for $j\neq l \in [k]$. 
The following quantity turns out to be the threshold at which the preference changes. 
\[ \forall j,l \in [k] :  j>l, \enspace D_{j,l} \defeq \frac{p_j-p_{l}} {d_j-d_{l}} \enspace. \] 
For convenience, we also define $D_{j,0} \defeq p_j/d_j$ for all $j\in [k]$. 
\begin{restatable}{lemma}{preferbundles}\label{lem:preferbundles}
	For all  $i\geq j>l \in [k]$, a buyer of type $(v,d_i)$ prefers a bundle of $d_j$ units to a bundle of $d_l$ units if and only if $v > D_{j,l}$.  Both bundles are equally preferable precisely when $v = D_{j,l}$. 
\end{restatable}
\begin{proof}
	The buyer prefers $d_j$ units over $d_l$ units if and only if
	$ v d_j - p_j  > v d_l - p_l .$
	Rearranging, we get the lemma. 
\end{proof}
 Before we proceed further, we note the following property for future reference. 
 \begin{restatable}{lemma}{dijorder}\label{lem:dijorder}
 	For all  $i\geq j\geq l \in [k]$, $D_{i,l}$ is a convex combination of (and hence is always in between) $D_{i,j}$ and $D_{j,l}$. 
 \end{restatable}
 \begin{proof}
 	It is easy to check the following identity. 
 	$  D_{i,l} = \frac{1}{d_i - d_l}\left((d_i - d_j )D_{i,j} + (d_j - d_l)D_{j,l} \right).$
 \end{proof}
 
We next consider how the optimum bundle changes for a given $d_i$, as $v$ decreases from $\bar V$ to 0. 
For high enough $v$, the optimum bundle for type $(v,d_i)$ should be $d_i$ units. 
As $v$ decreases, the optimal bundle is going to switch at the threshold  $\max_{j <i}\{D_{i,j} \} $ (to something in the $\arg \max$). Similarly, as $v$ decreases further, the optimal bundle is going to switch again and so on. 
In fact, these sequences for different $d_i$s are not independent and we can 
capture each such sequence of optimum bundles by 
 a single vector $\sigma\in \integers^k$ such that the $i^{\rm th}$ co-ordinate $\sigma(i) \in \arg \max_{j <i}\{D_{i,j}\}$.
Given such a $\sigma,$ for each $i$,  the sequence of optimal bundles for types with demand $d_i$ is given 
by the directed path $\path_\sigma(i)$, defined  as 
the (unique) longest path starting from $i$ in the directed graph on $[k]$ with edges $(i,\sigma(i))$. 
(The path ends when $\sigma(i) = 0$ for some $i$.)   

In fact, there is a closed form formula for the revenue function provided we know what the resulting $\sigma$ is. Towards this,  it is going to be more useful to consider the inverse of this map from $\pb$ to $\sigma$: 
given any  $\sigma\in \integers^k$ such that $\sigma(i ) \in [i-1]$, we define $\Delta_\sigma$ to be all the prices where the sequence of optimal bundles as described above is given by $\path_\sigma(i)$. 
Formally, 
$$ \Delta_{\sigma} \defeq \left\{ \pb: \forall i, \sigma(i) \in \arg \max_{j <i}\{D_{i,j} \}  \right\} .$$

\paragraph{Revenue function formula:}
We are now ready to give a closed form formula for the revenue function within each $\Delta_{\sigma}$. 
For ease of notation we let $F_i$ denote the conditional CDF $F_{d_i}$, and let $q_i$ to denote the probability that the buyer has a demand $d_i$. 
We also use $\sigma^2(i) $ to denote $\sigma(\sigma(i))$. 
We now define the following revenue function corresponding to $\sigma$ which captures $\Rev(\pb)$ in $\Delta_{\sigma}$: 
\[ \textstyle  \Rev_\sigma(\pb) \defeq  \sum_i  q_i\left(p_i \left(1 - F_i (D_{i,\sigma(i)}) \right) + \sum_{j \in \path_\sigma(i)} p_{\sigma(j)}  \left( F_i(D_{j,\sigma(j)}) - F_i (D_{\sigma(j),\sigma^2(j)}) \right)  \right)  ,\]

\begin{lemma}\label{lem:revenuedef}
$	\Rev(\pb)  = \Rev_\sigma(\pb) $ for all $\pb \in \Delta_\sigma$. 
\end{lemma}
\begin{proof}
	Suppose $\pb \in \Delta_\sigma$. Consider all buyer types with demand $d_i$. Among these, all types with value $v > D_{i,\sigma(i)}$ prefer to buy the bundle of $d_i$ units over any other bundle, 
	by Lemma~\ref{lem:preferbundles}, and because $\pb \in \Delta_\sigma$. These contribute $q_ip_i \left(1 - F_i (D_{i,\sigma(i)}) \right)$ to the revenue. 
	
		Now consider all types with value $v\in [D_{j,\sigma(j)}, D_{\sigma(j),\sigma^2(j)}]$ for some $j \in \path_\sigma(i)$. We need to prove that these prefer a bundle of $d_{\sigma(j)}$ over any other bundle $d_l$, so that they contribute to the revenue exactly $q_i p_{\sigma(j)}  \left( F_i(D_{j,\sigma(j)}) - F_i (D_{\sigma(j),\sigma^2(j)}) \right)  $, and the lemma follows. 		
		As characterized by Lemma~\ref{lem:preferbundles}, this follows from the following. 
			\begin{itemize}
				\item If $ l < \sigma(j)$, then $v\geq D_{\sigma(j),\sigma^2(j)}\geq D_{\sigma(j), l}$.  This holds because $\pb \in \Delta_\sigma$. 
				\item If $i \geq l > \sigma(j) $, then  $v \leq D_{j,\sigma(j)}\leq D_{l,\sigma(j)}$. We prove this in the rest of the proof.  
			\end{itemize}

We first prove that $\forall j \in \path_\sigma(i)$, 
	$	D_{j,\sigma(j)} \geq D_{\sigma(j), \sigma^2(j)} .$
	This follows from the fact that $	D_{j,\sigma^2(j)}$ is in between $	D_{j,\sigma(j)}$ and $ D_{\sigma(j), \sigma^2(j)}$  (Lemma~\ref{lem:dijorder}), 
	and that  $	D_{j,\sigma^2(j)}\leq 	D_{j,\sigma(j)}$ (since $\pb \in \Delta_\sigma$). 	
	We now prove the following: $\forall j \in \path_\sigma(i)$, and $l \in( \sigma(j),j]$, we have that 
$		D_{l, \sigma(j)} \geq D_{j, \sigma(j)}.$
	This follows from the fact that if $l \in (\sigma(j),j ]$, then $D_{j,\sigma(j)}$ is in between $D_{j,l}$ and $D_{l,\sigma(j)}$ (from Lemma~\ref{lem:dijorder}), and $D_{j,l}\leq D_{j,\sigma(j)}$. 
	Now by a repeated application of 
	the fact $	D_{j,\sigma(j)} \geq D_{\sigma(j), \sigma^2(j)}, $ 
	we get the same conclusion for all  $j$ and $l$ such that $i \geq l > \sigma(j)$. 
\end{proof}

\paragraph{Concavity of~ $\Rev_\sigma$:}
We next show that each of the $\Rev_\sigma$s by itself is a concave function. We do this by showing that $\Rev_\sigma$ can be written as a positive linear combination of linear functions, and compositions of the functions $v (1- F_d(v))$ with linear functions. Since the $v(1 -F_d(v))$ functions are concave by assumption, and such compositions and positive linear combinations preserve concavity,  $\Rev_\sigma$ is concave too. 
\begin{restatable}{lemma}{revsigmaconcave}\label{lem:revsigmaconcave}
	For all $\sigma$, $\Rev_\sigma(\pb)$ is a concave function. 
\end{restatable}
\begin{proof}
	We can rewrite $\Rev_\sigma$ as follows, using the definition of $D_{j,l}$. 
\[ \textstyle \Rev_\sigma =  \sum_i  q_i\left(p_i -  \sum_{j \in \path_\sigma(i)}F_i(D_{j,\sigma(j)})   \left( p_j - p_{\sigma(j)} \right)  \right) \]
\[ \textstyle = \sum_i  q_i\left(p_i -  \sum_{j \in \path_\sigma(i)}   F_i(D_{j,\sigma(j)})  D_{j,\sigma(j)} \left(j - \sigma(j)\right)   \right).\]
	We assumed that $v (1-F_i(v))$ is concave, which implies that $-vF_i(v)$ is concave. 
	$D_{j,\sigma(j)}$ is a linear function of $ \pb$ for all $j$. 
	Since composition of linear  functions with concave functions is concave, it follows that 
	$-  F_i(D_{j,\sigma(j)})  D_{j,\sigma(j)}  $
	is concave. Now $\Rev_\sigma$ is a positive linear combination of concave functions, which makes it concave too. 
\end{proof}

\paragraph{Stitching the $\Rev_\sigma$s together:}
Lemmas~\ref{lem:revenuedef} and \ref{lem:revsigmaconcave} imply that $\Rev$ is  piecewise concave, i.e., inside each $\Delta_\sigma$ it is concave. 
In general this does not imply that such a function is concave everywhere. 
One property that would imply that $\Rev$ is concave everywhere would be if $\Rev$ was equal to $\min_\sigma \Rev_\sigma$. 
Unfortunately, this is not true. In fact, there is a partial order over $\sigma$s that determine when one $\Rev_\sigma$ is always greater than the other. 
We show a different, and somewhat surprising,  property of the $\Rev_\sigma$s that also implies that $\Rev$ is concave. 
We show that at the boundaries between two regions not only do the corresponding $\Rev_\sigma$s agree (which they should, for $\Rev$ to be even continuous), 
but also their gradients agree! 

\begin{lemma}\label{lem:boundary}
	For all $\sigma, \sigma'$, $\pb$ such that $\pb \in \Delta_\sigma \cap \Delta_{\sigma'}$, we have that 
	\[ \Rev_\sigma(\pb) = \Rev_{\sigma'} (\pb) \text{ and }  \nabla \Rev_\sigma(\pb) = \nabla \Rev_{\sigma'} (\pb).  \]
\end{lemma}

\begin{proof}
	We first argue that it is sufficient to prove Lemma \ref{lem:boundary} for the case where $\sigma$ and $\sigma'$ disagree 
in exactly one co-ordinate, i.e., there is some $\icrit$ such that $\sigma(\icrit) \neq \sigma'(\icrit)$, and $\forall j \neq \icrit$, $\sigma(j) = \sigma'(j)$. Suppose we have done that. Now consider any two $\sigma$ and $\sigma'$, and 
a sequence $\sigma = \sigma_1,\sigma_2,\ldots, \sigma_n = \sigma'$ such that 
for any $i$, $\sigma_i $ and $\sigma_{i+1}$ differ in exactly one co-ordinate, where $\sigma_i$ agrees with 
$\sigma$ in that co-ordinate and $\sigma_{i+1}$ agrees with $\sigma'$. 
The fact that $\pb \in \Delta_\sigma \cap \Delta_{\sigma'}$  implies that for all co-ordinates $j$ such that $\sigma(j) \neq \sigma'(j)$, 
both $\sigma(j)$ and $\sigma'(j) \in \arg \max _{j'<j}\{D_{j,j'} \}.$
Similarly, $\pb \in \Delta_{\sigma_i} \cap \Delta_{\sigma_{i+1}}$ requires the same condition, but only for the co-ordinate that they differ in, and therefore $\pb \in \cap_{i=1}^n \Delta_{\sigma_i}$. 
Since we know Lemma \ref{lem:boundary} holds when the two $\sigma$s differ in at most one co-ordinate, it now follows that 
$\Rev$ and $\nabla \Rev$ at $\pb$ are the same for all $\sigma_i$s and hence for $\sigma$ and $\sigma'$ as well. 

Now we prove Lemma \ref{lem:boundary} when $\sigma$ and $\sigma'$ differ at exactly one co-ordinate, $\icrit$. 
We consider the portions of the paths $\path_\sigma(i) $ and $\path_{\sigma'}(i) $ that are disjoint, and refer to these 
disjoint portions as simply $\path\subseteq \path_\sigma(i)$ and $\path'\subseteq\path_{\sigma'}(i) $. 
Both of these paths start at $\icrit$ and end at $\iend$. 
Note that once the two paths merge, they remain the same for the rest of the way.  
If the paths don't merge, then we let $\iend = 0$. 
The critical fact we use is that along these paths the $D$s are all the same, which is stated as the following lemma. 


\begin{claim}\label{lem:zero_on_path}
All $j, j' \in \path \cup \path'$ s.t. $j> j'$ have the same $D_{j,j'}$. 
\end{claim}
\begin{proof}
We prove the claim by induction, where we add one node at a time in the following order. 
We start the base case with $\icrit, \sigma(\icrit)$ and $\sigma'(\icrit)$. 
At any point let $j$ and $j'$ be the last points on $\path$ and $\path'$ that we have added so far. 
In the inductive step, if $j> j'$, we add $\sigma(j)$ and otherwise we add $\sigma'(j')$. 
We stop when all nodes in $\path \cup \path'$ have been added.

For the base case, let $j = \sigma(\icrit)$ and $j' = \sigma'(\icrit)$. Without loss of generality, assume that $j > j'$. By \autoref{lem:dijorder}, we get that $D_{\icrit,j'}$ is between $D_{\icrit,j}$ and $D_{j,j'}$. Since  $D_{\icrit,j} = D_{\icrit,j'}$, from the definition of $\icrit$, we get $D_{\icrit,j} = D_{j,j'} = D_{\icrit,j'}$.

For the inductive step, let $j \in \path$ and $j' \in \path' $ be the last points that we have added so far, and again without loss of generality $j > j'$. 
 Let $v = \sigma(j)$. If $v=j'$ we are done. There are two cases: $v > j'$ and $v<j'$. 
 In the former case, we have $D_{j,v} \geq D_{j,j'}$ from the definition of $v$. 
From \autoref{lem:dijorder}, $D_{j,j'}$ must be in between $D_{j,v}$ and $D_{v,j'}$, therefore $D_{j,j'}\geq D_{v,j'}$. 
Let $i'\in \path'$ be the predecessor of $j'$, i.e., $\sigma'(i') = j'$. 
Due to the order in which we added the nodes, it must be that $i'> j$.  
By definition, $D_{i',j'} \geq D_{i',v}$, and by \autoref{lem:dijorder} $D_{i',j'}$ must be in between 
$D_{i',v} $ and $D_{v,j'}$, therefore $D_{v,j'}\geq D_{i',j'}$.
By the inductive hypothesis, we have that $D_{i',j'} = D_{j,j'} $ and hence they both must be equal to $D_{v,j'}$. 

Now consider any $i\neq j'$ that we have already added. It must be that $i <v$, and hence $D_{i,j'}$ must be in between 
$D_{i,v}$ and $D_{v,j'}$, but from the argument in the previous paragraph and the inductive hypothesis, we have 
that $D_{i,j'} = D_{v,j'}$ , and hence they must be equal to $D_{i,v}$. This completes the induction for this case. 
%
%
 The latter case of $v < j'$ is identical.
\end{proof}

\paragraph{Continuing the proof of Lemma~\ref{lem:boundary}:} 	

To show that $\Rev_\sigma$s agree on the boundary, consider the difference $\Rev_\sigma(\pb) - \Rev_{\sigma'} (\pb)$. For all $i \leq \icrit$, or $i$ such that $\icrit \notin \path_{\sigma}(i)$, nothing changes, therefore all those terms cancel out. Moreover, even for $i$ such that $\icrit \in \path_{\sigma}(i)$, the only terms that don't cancel out are $j \in \path \cup \path'$. Therefore, we get:
\[ \textstyle
\Rev_\sigma(\pb) - \Rev_{\sigma'} (\pb) 
= \sum_{i \geq \icrit: \icrit \in \path_\sigma(i)} q_i \left( \sum_{j \in \path} p_{\sigma(j)}  \left( F_i(D_{j,\sigma(j)}) - F_i (D_{\sigma(j),\sigma^2(j)}) \right) \right. \]
\[  \textstyle
\left. \qquad - \sum_{j \in \path'} p_{\sigma'(j)}  \left( F_i(D_{j,\sigma'(j)}) - F_i (D_{\sigma'(j),(\sigma')^2(j)}) \right) \right),
\] 
which is zero by Claim~\ref{lem:zero_on_path}.

For the second part of the proof, we'll show that the gradient of $\Rev_\sigma - \Rev_{\sigma'}$ is zero. 
We only need to consider the partial derivatives w.r.t. $p_j$ for $j \in \path \cup \path'$ (modulo some corner cases). 
%
 Fix a $j \in \path$, and  consider the terms in $\frac{\partial (\Rev_\sigma - \Rev_{\sigma'})}{\partial p_j}$ corresponding to 
 some $i \geq \icrit$ such that $\icrit \in \path_{\sigma}(i)$, in the outer summation. 
 Let the path $\path_{\sigma}(i)$ be such that $a \in \path_{\sigma}(i)$, $b = \sigma(a)$, $j = \sigma(b)$, $c = \sigma(j)$ and $d = \sigma(c)$. 
\[ i \rightarrow \ldots \rightarrow \icrit \rightarrow \ldots \rightarrow a \rightarrow b \rightarrow j \rightarrow c \rightarrow d \rightarrow \ldots \]
Then the terms under consideration are 
\begin{align*}
& \frac{\partial}{\partial p_j} q_i \left(  p_{b}  \left( F_i(D_{a,b}) - F_i(D_{b,j}) \right) +  p_{j}  \left( F_i(D_{b,j}) - F_i(D_{j,c}) \right) + p_{c}  \left( F_i(D_{j,c}) - F_i(D_{c,d}) \right)  \right) = \\
&= q_i \left( \frac{p_b}{d_b-d_j} f_i(D_{b,j}) - \frac{p_j}{d_b-d_j} f_i(D_{b,j}) + F_i(D_{b,j}) - F_i(D_{j,c})- \frac{p_j}{d_j-d_c} f_i(D_{j,c}) + \frac{p_c}{d_j-d_c} f_i(D_{j,c}) \right) \\
&= q_i \left( D_{b,j} f_i(D_{b,j}) - D_{j,c} f_i(D_{j,c}) + F_i(D_{b,j}) - F_i(D_{j,c}) \right).
\end{align*}
By Claim~\ref{lem:zero_on_path}, $D_{b,j} = D_{j,c}$, and therefore these terms are zero. The cases when $i = \icrit$, or $\icrit = a,b,j$, or $c,d=0$, or $j \in \path_{\sigma'}(i)$ are identical. \end{proof}

We are now ready to prove the main theorem of this section, which is simply arguing how this agreement of gradients implies that $\Rev$ is concave everywhere. 
\begin{proof}[Proof of Theorem~\ref{thm:concavity}]
	Consider any two prices $\pb_1$ and $\pb_2$, and the line segment joining the two. We will argue that $\Rev$ is concave along this line segment, which then implies the Theorem. 
	From  Lemmas~\ref{lem:revenuedef} and~\ref{lem:revsigmaconcave}, we have that this line segment is itself divided into many intervals, and within each interval, 
	$\Rev$ is a concave function. Further, from Lemma~\ref{lem:boundary}, we have that these concave functions agree at the intersections of the intervals, and the 
	gradients agree too. Thus $\Rev$ is smooth, and the derivative along this line is monotone. This implies that $\Rev$ is concave along the line. 
\end{proof}

\bibliographystyle{plainnat}
\bibliography{references}

\appendix
\section{Deferred Proofs}\label{app:appendix}
The proof below contains the calculations for Example~\ref{ex:regularity comparisons}.
\begin{proof}
We first show that the constant elasticity distribution with cumulative density $F(v) = 1 - (v/a)^{1/\epsilon}$ is DMR.  Recall that DMR is equivalent to concavity of the revenue function.  To verify concavity, we calculate the second derivate of the revenue function and show that it is negative.
\begin{align*}
R'(v) &= (\frac{v}{a})^{1/\epsilon} + \frac{v}{a\epsilon}(\frac{v}{a})^{1/\epsilon -1}.\\
R''(v) &= (\frac{v}{a})^{1/\epsilon -1}\frac{2}{a\epsilon} + (\frac{v}{a})^{1/\epsilon -2}\frac{v}{a^2\epsilon}(1/\epsilon -1) \\
& = (\frac{v}{a})^{1/\epsilon -2}(\frac{2v}{a^2\epsilon} + \frac{v}{a^2\epsilon}(1/\epsilon - 1)) \\
& = (\frac{v}{a})^{1/\epsilon -2}\frac{v}{a^2\epsilon}(1+1/\epsilon) \leq 0.
\end{align*}

Now consider regularity.  Note that the probability density function $f(v) = \frac{-1}{\epsilon a} (v/a)^{1/\epsilon - 1}$.  Recall that a distribution is regular if the function $\phi(v)$ is monotone non-decreasing in $v$.
\begin{align*}
\phi(v) &= v - \frac{1-F(v)}{f(v)}\\
&= v - \frac{(v/a)^{1/\epsilon}}{\frac{-1}{a\epsilon}(v/a)^{1/\epsilon-1}}\\
&= v - \frac{v/a}{-1/(a\epsilon)} = v (1+\epsilon),
\end{align*}

which is monotone \emph{decreasing} since by assumption $\epsilon < -1$.

We finally argue that the exponential distribution, defined as $F(v) = 1 - e^{-v}$ is not DMR but is regular.  The revenue function is $R(v) = ve^{-v}$, its first derivative is $R'(v) = (1-v)e^{-v}$, and its second derivative is $(v-2)e^{-v}$, which is positive for $v \geq 2$, violating concavity.  However, as commonly known, this distribution is regular since $\phi(v) = v - \frac{1-F(v)}{f(v)} = v- \frac{e^{-v}}{e^{-v}} = v -1$ is monotone non-decreasing in $v$.
\end{proof}

\expost*

\begin{proof}
	Consider an EIC and EIR mechanism.  First note that we can assume that for each type $(v,d)$, the randomized allocation $A(v,d)$ does not assign a number of units more than $d$.  If this is not true, replace any assignment of more than $d$ units with the assignment of $d$ units.  Note that this change does not  change the utility of truthful reporting, and cannot improve utility of non-truthful reporting.  Therefore the resulting mechanism is EIC and EIR.  Now consider a type $(v,d)$, its \emph{realized} allocation $A(v,d)$, and its expected payment $p(v,d)$, and construct a randomized payment $\tilde{p}(v,d)$ as follows
	\begin{align*}
	\tilde{p}(v,d) = \frac{p(v,d) A(v,d)}{\E \left[ A(v,d)\right]}.
	\end{align*}
	\noindent Note that the expected payment of the type stays the same,
	\begin{align*}
	 \E \left[ \tilde{p}(v,d)\right] = p(v,d) \frac{\E \left[ A(v,d)\right]}{\E \left[ A(v,d)\right]} = p(v,d).
	\end{align*}  
	\noindent As a result, the modified mechanism stays EIC.  In addition, the ex-post utility of the type from the realized allocation of $A(v,d)$ units is 
	\begin{align*}
	vA(v,d)- \frac{p(v,d)A(v,d)}{\E \left[ A(v,d)\right]} ,
	\end{align*}
	which is non-negative if and only if 
	\begin{align*}
	v \E \left[ A(v,d)\right] - p(v,d) \geq 0,
	\end{align*}
	which hold by EIR.
\end{proof}

\Support*


\begin{proof}
	Let $x^i_m= \Prob [A(t_i) \geq m]$ denote the probability that type $t_i$ is allocated $m$ or more units.
	Let $h_i$ be such that $x^i_{h_i} > 0$ and $x^i_{h_i+1} = 0$. Set $w\left(t_i\right) = \frac{\sum_{m=1}^{h_i} x^i_m}{d_i}$, and 
	consider an alternate allocation given by a random variable $B(t_i)$  
	which is $d_i $ with probability $w(t_i)$ and is 0 otherwise.  
	Then, $\Prob [B (t_i) \geq m] = y^i_m = w\left(t_i\right)$ for all $m \leq d_i$. 
	The utility of $t_i$ when reporting $t_i$ remains unchanged under this alternate allocation: 
	\[ \ut{t_i}{t_i} = v_i  \sum_{m=1}^{d_i} x^i_m - p\left( t_i \right) = v_i d_i w\left(t_i\right) - p\left( t_i \right) = v_i  \sum_{m=1}^{d_i} y^i_m - p\left( t_i \right)\]
	and so does $\ut{t_j}{t_i}$ for all $t_j$ with $d_j \geq d_i$. When $d_j < d_i$, it is easy to check that $\sum_{m=1}^{d_j} x^i_m \leq \sum_{m=1}^{d_j} y^i_m = d_j w\left(t_i\right)$, since $x^i_m \geq x^i_{m+1}$ for all $k$, and thus the utility of $t_j$ when reporting $t_i$ can only decrease.
	Thus, when changing the allocation from $A$ to $B$, the EIC constraints are still satisfied, and total revenue remains unchanged.
\end{proof}

\utilandX*

\begin{proof}
	$\displaystyle u(v,d_i \rightarrow v \frac{d_{i}}{d_j}, d_{j} ) = v d_i w(v \frac{d_{i}}{d_{j}}, d_{j}) - p(v \frac{d_{i}}{d_{j}}, d_{j}) $\\
	$\displaystyle = (v \frac{d_i}{d_{j}}) d_{j}  w(v \frac{d_{i}}{d_{j}}, d_{j}) - p(v \frac{d_{i}}{d_{j}}, d_{j}) = u(v \frac{d_{i}}{d_{j}}, d_{j} \rightarrow v \frac{d_{i}}{d_{j}}, d_{j}  ) = d_{j} U_{d_{j}} (v \frac{d_i}{d_{j}})$.\\
	The second part is identical.
\end{proof}

\localtoglobal*

\begin{proof}
We will show that for all pairs of types $t_i =(v_i,d_i)$ and $t_j = (v_j,d_j)$, with $d_i \geq d_j + 1$, $t_i$ does not want to report $t_j$ and vice versa:
\begin{itemize}
\item $\displaystyle \ut{t_i}{t_i} \geq \ut{v_i,d_i}{v_i,d_i - 1} = \ut{v_i,d_i - 1}{v_i,d_i - 1}$ \\
$\qquad \geq \ut{v_i,d_i-1}{v_i,d_i - 2} = \ut{v_i,d_i-2}{v_i,d_i - 2} $\\
$\qquad \dots  $\\
$\qquad \geq \ut{v_i,d_j}{v_i,d_j} \geq  \ut{v_i,d_j}{v_j, d_j} = \ut{v_i,d_i}{v_j, d_j} = \ut{t_i}{t_j}$

\item $\ut{t_j}{t_j} \geq \ut{t_j}{v_j \frac{d_j}{d_j + 1},d_j+1} = v_j d_j \w{v_j \frac{d_j}{d_j + 1},d_j+1} - \p{v_j \frac{d_j}{d_j + 1},d_j+1}$ \\
$= v_j \frac{d_j}{d_j+1} (d_j+1) \w{v_j \frac{d_j}{d_j + 1},d_j+1} - \p{v_j \frac{d_j}{d_j + 1},d_j+1}$ \\
$= \ut{v_j \frac{d_j}{d_j + 1},d_j+1}{v_j \frac{d_j}{d_j + 1},d_j+1}$

Applying this argument repeatedly gives us: $\ut{t_j}{t_j} \geq \ut{v_j \frac{d_j}{d_i},d_i}{v_j \frac{d_j}{d_i},d_i}.$ Using truthfulness for a fixed $d$, we have that the RHS is at least $\ut{v_j \frac{d_j}{d_i},d_i}{v_i,d_i} = v_j \frac{d_j}{d_i} d_i \w{v_j \frac{d_j}{d_i},d_i} - \p{v_j \frac{d_j}{d_i},d_i}$, which is just $\ut{t_j}{t_i}$.
\end{itemize}
\end{proof}

\section{Detailed Characterization for $\total = 2$}
\label{app:detailed d2}

We first complete the case analysis that shows that the optimal mechanism is deterministic for $\total = 2$.

We have the following cases for $v_1$ and $v_2$:
\begin{description}
	\item {$v_1 \leq v_2$:}
	Thus, for all $v \leq v_1 \leq v_2$, we have that $ U_{d_1}(v) = \frac{d_2}{d_1} U_{d_2}\left( v \frac{d_1}{d_2} \right) = U_{d_1} \left( v \frac{d_1}{d_2} \right),$ which implies that $U_{d_1}(v) = 0$, and therefore $U_{d_2}(v) = 0$. For all $v \geq v_1$, we have that $U_{d_1}'(v)=1$, i.e. $\w[d_1]{v} = 1$; 
	this corresponds to a posted price of $d_1 v_1$ for a bundle of $d_1$ units. For $v \in \left[ v_1 , v_2 \right]$, we have $ U_{d_2}(v) = \frac{d_1}{d_2} U_{d_1}(v) = \frac{d_1}{d_2}(v-v_1)$. This implies that the allocation function $\w[d_2]{v}$ is equal to $\frac{d_1}{d_2}$; for a price of $d_1 v_1$, we offer a bundle of $d_2$ units with probability $\frac{d_1}{d_2}$. For $v \geq v_2$, we have a posted price of $d_2 v_2 - (v_2 - v_1)d_1$ for a bundle of $d_2$ units. The same allocation rule can be induced by just two menu units (and no randomization): $d_1$ units cost $v_1 d_1$ and $d_2$ units cost $v_2 d_2 - (v_2 - v_1)d_1$.
	\item {$v_2 \leq v_1$ and $v_1 d_1 \leq v_2 d_2$:}
	As before, for all $v \leq v_2$, $U_{d_2}(v) = U_{d_1}(v) = 0$, and the bundle of $d_2$ units has a posted price of $v_2 d_2$. For $v \in \left[ v_2 , v_1 \right]$, we have $ U_{d_1}(v) = \frac{d_2}{d_1} U_{d_2} \left( v \frac{d_1}{d_2} \right) \leq \frac{d_2}{d_1} U_{d_2} \left( v_2 \right) = 0$. For $v \geq v_1$, $U'_{d_1}(v) = 1$; this is a posted price of $d_1 v_1$ for $d_1$ units.
	\item {$v_2 \leq v_1$ and $v_1 d_1 > v_2 d_2$:} 
	Once again, for all $v \leq v_2$, $U_{d_2}(v) = U_{d_1}(v) = 0$, and the bundle of $d_2$ units has a posted price of $d_2 v_2$. For $v \in \left[ v_2 , \frac{d_2}{d_1} v_2 \right]$, we have $U_{d_1}(v) = \frac{d_2}{d_1} U_{d_2} \left( v \frac{d_1}{d_2} \right) = 0$. For $v \in \left[ \frac{d_2}{d_1} v_2 , v_1 \right]$, $U'_{d_1}(v) = 1$; offer a bundle of $d_1$ units for a price of  $d_1 \frac{d_2}{d_1} v_2 = d_2 v_2$. This corresponds to selling only the $d_2$ bundle for a price of $v_2d_2$. 
\end{description}

We now characterize the optimal thresholds $v_1$ and $v_2$.  Let $v_1$ and $v_2$ be the values after which $(.,d_1)$ and $(.,d_2)$ type agents are allocated $d_1$ and $d_2$ units respectively. Then, the optimal mechanism posts a price for $d_1$ units and a price for $d_2$ units that is either: (1) $v_1 d_1$ and $v_2 d_2 - (v_2 - v_1)d_1$, (2) $v_1 d_1$ and $v_2 d_2$, or (3) $v_2 d_2$ for both. This is equivalent to the maximum of:

\begin{itemize}
\item $\max v_1 d_1 \left( 2 - F_1(v_1) - F_2(v_1) \right) + v_2 (d_2 - d_1) \left( 1 - F_2(v_2) \right)$\\
subject to $v_2 \geq v_1$.

\item $\max v_1 d_1 \left( 1 - F_1(v_1) \right) + v_2 d_2 \left( 1 - F_2(v_2) \right)$\\
subject to $v_1 \geq v_2$ and $v_1 \leq \frac{d_2}{d_1} v_2$.

\item $\max v_2 d_2 \left( 2 - F_2(v_2) - F_1(v_2 \frac{d_2}{d_1}) \right)$.

\end{itemize}

Let $\vstar{1}$ and $\vstar{2}$ be the optimal choices for $v_1$ and $v_2$. Also, let $\hat{v}_1$ and $\hat{v}_2$ be the monopoly pricing solutions, i.e. $\hat{v_i} = arg\max v d_i \left( 1 - F_i(v) \right)$. 

Then, we have the following options for $\vstar{1}$ and $\vstar{2}$: 
\begin{enumerate}
\item $\vstar{1} = \hat{v}_1$ and $\vstar{2} = \hat{v}_2$ (unconstrained version of the second bullet)
\item $\vstar{1} = arg\max v d_1 \left( 2 - F_1(v) - F_2(v) \right)$ and $\vstar{2} = \hat{v}_2$ (unconstrained version of the first bullet)
\item $\vstar{1} = \vstar{2} = arg\max v \left( d_1 (1 - F_1(v)) + d_2 (1-F_2(v)) \right)$ 
\item $\frac{d_1}{d_2} \vstar{1} = \vstar{2} = arg\max v d_2 \left( 2 - F_2(v) - F_1(v \frac{d_2}{d_1}) \right)$
\end{enumerate}

This corresponds to the following: compute $\hat{v}_1$ and $\hat{v}_2$. If $\frac{d_2}{d_1} \hat{v}_2 \geq \hat{v}_1 \geq \hat{v}_2$ we're done. Otherwise, compute $arg\max v d_1 \left( 2 - F_1(v) - F_2(v) \right)$. If it is at most $\hat{v}_2$, then pick the best option out of 2,3 and 4. If not, pick the best out of 3 and 4.

\begin{itemize}
\item $\max v_1 d_1 \left( 2 - F_1(v_1) - F_2(v_1) \right) + v_2 (d_2 - d_1) \left( 1 - F_2(v_2) \right)$\\
subject to $v_2 \geq v_1$.

\item $\max v_1 d_1 \left( 1 - F_1(v_1) \right) + v_2 d_2 \left( 1 - F_2(v_2) \right)$\\
subject to $v_1 \geq v_2$ and $v_1 \leq \frac{d_2}{d_1} v_2$.

\item $\max v_2 d_2 \left( 2 - F_2(v_2) - F_1(v_2 \frac{d_2}{d_1}) \right)$
\end{itemize}

Let $\vstar{1}$ and $\vstar{2}$ be the optimal choices for $v_1$ and $v_2$. Also, let $\hat{v}_1$ and $\hat{v}_2$ be the monopoly pricing solutions, i.e. $\hat{v_i} = arg\max v d_i \left( 1 - F_i(v) \right)$. The following procedure gives the optimal $\vstar{1}$ and $\vstar{2}$: compute $\hat{v}_1$ and $\hat{v}_2$, and check whether they satisfy the IC constraints. If they do, then we are done. If they do not, it must be that either $\hat{v_1} < \hat{v_2}$, or $\hat{v_1} > \frac{d_2}{d_1} \hat{v_2}$.

In the former case, compute the \textit{best per unit price} $q$, i.e. a price $q$ such that $d_1$ units cost $qd_1$ and $d_2$ units cost $qd_2$. This corresponds to the solution of the first bullet. 

In the latter case, compute the \textit{best bundle price}, i.e. the best price $p$ that is going to be the same for $d_1$ and $d_2$. This corresponds to the solution of the third bullet. The best of $p$ and $q$ the two is optimal, and given that, $\vstar{1}$ and $\vstar{2}$ can be easily calculated.

Then, we have the following options for $\vstar{1}$ and $\vstar{2}$: 
\begin{enumerate}
\item $\vstar{1} = \hat{v}_1$ and $\vstar{2} = \hat{v}_2$ (unconstrained version of the second bullet)
\item $\vstar{1} = arg\max v d_1 \left( 2 - F_1(v) - F_2(v) \right)$ and $\vstar{2} = \hat{v}_2$ (unconstrained version of the first bullet)
\item $\vstar{1} = \vstar{2} = arg\max v \left( d_1 (1 - F_1(v)) + d_2 (1-F_2(v)) \right)$ 
\item $\frac{d_1}{d_2} \vstar{1} = \vstar{2} = arg\max v d_2 \left( 2 - F_2(v) - F_1(v \frac{d_2}{d_1}) \right)$
\end{enumerate}

This corresponds to the following: compute $\hat{v}_1$ and $\hat{v}_2$. If $\frac{d_2}{d_1} \hat{v}_2 \geq \hat{v}_1 \geq \hat{v}_2$ we're done. Otherwise, compute $arg\max v d_1 \left( 2 - F_1(v) - F_2(v) \right)$. If it is at most $\hat{v}_2$, then pick the best option out of 2,3 and 4. If not, pick the best out of 3 and 4.

\end{document}